\PassOptionsToPackage{hyphens}{url}
\documentclass[ninept]{article}

\usepackage[preprint]{spconf}
\copyrightnotice{\copyright\ IEEE 2021 }
\toappear{Published as a conference paper at {\it ICASSP 2021.}}

\usepackage{amsmath, amssymb}
\usepackage{bm}
\usepackage{cite}

\usepackage[hyphens]{url}

\usepackage{mathtools}

\usepackage{mathrsfs}
\usepackage{algorithm}
\usepackage{algorithmicx}
\usepackage{algpseudocode}

\usepackage{subcaption}
\usepackage[skip=3pt]{caption}

\newcommand\logsumexp[1]{\log\sum_{#1}\exp}
\DeclareMathOperator*{\mybigoplus}{\bigoplus\nolimits_\gamma}
\newcommand{\sisdr}{\text{SI-SDR}}
\newcommand{\prob}{p}
\newcommand{\numsrc}{N}
\newcommand{\nummix}{M}
\newcommand{\bigo}{O}

\newcommand{\perms}{\mathcal{P}_\numsrc}
\newcommand{\birkhoff}{\mathcal{B}_\numsrc}
\newcommand{\entropy}[1]{H(#1)}

\usepackage{amsthm}
\newtheorem{theorem}{Theorem}
\newtheorem*{theorem*}{Theorem}
\newtheorem*{problem*}{Problem}
\newtheorem*{pitloss*}{PIT Loss}
\newtheorem*{dfn*}{Definition}

\newcommand\mat[1]{\mathbf{{#1}}}
\newcommand\frob[2]{\left\langle{#1},{#2}\right\rangle_\mathsf{F}}

\usepackage{enumerate}

\newcommand{\appropto}{\mathrel{\vcenter{
  \offinterlineskip\halign{\hfil$##$\cr
    \propto\cr\noalign{\kern2pt}\sim\cr\noalign{\kern-2pt}}}}}

\usepackage{tikz}
\newcommand\copyrighttext{%
    \footnotesize \copyright 2021 IEEE. Personal use of this material is permitted.
    Permission from IEEE must be obtained for all other uses, in any current or future
    media, including reprinting/republishing this material for advertising or promotional
    purposes, creating new collective works, for resale or redistribution to servers or
    lists, or reuse of any copyrighted component of this work in other works.
}\newcommand\mycopyrightnotice{%
\begin{tikzpicture}[remember picture,overlay]
\node[anchor=north,yshift=-10pt] at 
    (current page.north) {{\parbox{\dimexpr\textwidth-\fboxsep-\fboxrule\relax}{\copyrighttext}}};
\end{tikzpicture}%
}

\title{
    Towards Listening to 10 People Simultaneously:
    An Efficient Permutation Invariant Training of Audio Source Separation Using Sinkhorn's Algorithm
}%
\name{Hideyuki Tachibana
}
\address{PKSHA Technology Inc., Hongo, Bunkyo, Tokyo, Japan}
\begin{document}
\ninept
\maketitle
\begin{abstract}
    In neural network-based monaural speech separation techniques,
    it has been recently common to evaluate the loss using the permutation invariant training (PIT) loss.
    However, the ordinary PIT requires to try all $\numsrc!$ permutations
    between $\numsrc$ ground truths and $\numsrc$ estimates.
    Since the factorial complexity explodes very rapidly as $\numsrc$ increases,
    a PIT-based training works only when the number of source signals is small, 
    such as $\numsrc = 2$ or $3$.
    To overcome this limitation,
    this paper proposes a SinkPIT, a novel variant of the PIT losses,
    which is much more efficient than the ordinary PIT loss when $\numsrc$ is large.
    The SinkPIT is based on Sinkhorn's matrix balancing algorithm,
    which efficiently finds 
    a doubly stochastic matrix which approximates the best permutation
    in a differentiable manner.
    The author conducted an experiment
    to train a neural network model to decompose a single-channel mixture into 10 sources using the SinkPIT,
    and obtained promising results.
\end{abstract}
\mycopyrightnotice%
\begin{keywords}
    audio source separation,
    permutation invariant training,
    Sinkhorn algorithm,
    neural networks.
\end{keywords}

\section{Introduction}
    Audio source separation
    is a versatile module for many practical audio applications
    e.g.\ a speech recognition in a crowded situation, a hearing aid, a meeting transcription, etc. 
    The objective of audio source separation is to solve an inverse problem as follows.
    \begin{problem*}
        There are $\numsrc$ unknown source signals $[s_1(t), \cdots, s_\numsrc(t)]$ 
        $($e.g.\ speech signals of $\numsrc$ people talking simultaneously$)$,
        and we are given $\nummix$ mixtures of them $x_i(t) = \sum_{j=1}^\numsrc A_{ij} s_j(t), (1 \le i \le \nummix)$, 
        where $\mat{A} = (A_{ij}) \in \mathbb{R}^{\nummix \times \numsrc}$ is the mixing matrix, which is also unknown.
        Then, construct a function
        $f_\theta: [x_1(t), \cdots, x_\nummix(t)] \mapsto [y_1(t), \cdots, y_\numsrc(t)]$ 
        such that the set of estimates $\{y_i(t)\}_{i=1}^N$ approximates the set of source signals $\{s_i(t)\}_{i=1}^N$.
        To this end, we may use domain knowledge,
        or training data $($a collection of speech signals$)$,
        or other limited auxiliary $($e.g. spatial$)$ information.
    \end{problem*}
    The underdetermined case $\nummix < \numsrc$ is especially challenging,
    and has been the subject of many attempts.
    (Early pioneering work includes~\cite{lee1999blind,amari1999natural,bofill2001underdetermined} etc.)
    Recently, the single-mixture source separation problem (i.e.\ $\nummix = 1, \numsrc \ge 2$)
    which does not require special hardware such as microphone arrays, 
    has received much attention,
    and a number of techniques based 
    on deep neural networks~\cite{Isik2016sep,%
        yu2017permutation,kolbaek2017multitalker,kinoshita2018listening,luo2018tasnet,%
        takahashi2019recursive,Luo_2019,yousefi2019probabilistic,nachmani2020voice,luo2020dual,chen2020dual,
        wisdom2020unsupervised}
    have been proposed.
    In particular, ConvTasNet~\cite{Luo_2019} brought a significant breakthrough 
    in that it outperformed the ideal binary/ratio masks 
    (IBM/IRM)~\cite{wang1999separation,wang2014training}
    which had been implicitly supposed to be the ``upper limits'' or the ``targets'' of source separation techniques.

    Despite the powerful capabilities of deep neural networks,
    however,
    most studies have focused on improving the SDRs, 
    fixing the number of source signals $\numsrc$ to be $2$ or $3$.
    To our knowledge, $\numsrc=5$ is the largest $\numsrc$ reported~\cite{nachmani2020voice}.
    In this paper, we take advantage of its powerful potential to separate a single mixture into much more sources.
    In specific, we consider an extreme case $\numsrc=10$..

    A reason why the studies in this direction have attracted less attention would be the 
    difficulty to solve the permutation ambiguity problem.
    That is, as 
    the indices of signals are arbitrary,
    we need to evaluate all possible $s_i$-$y_j$ permutations
    to evaluate the loss between $\{s_i(t)\}$ and $\{y_j(t)\}$.
    Recent neural network-based studies refer to
    this process as the Permutation Invariant Training (PIT)~\cite{Isik2016sep,yu2017permutation,kolbaek2017multitalker}.
    \begin{dfn*}[PIT Loss~\cite{Isik2016sep,yu2017permutation,kolbaek2017multitalker}]
        Given a pairwise loss matrix $\mat{C} \in \mathbb{R}^{\numsrc \times \numsrc}$
        whose each element $C_{ij}$ is the pairwise loss between $s_i(t)$ and $y_j(t)$.
        Let $\sigma : i \mapsto j$ be a permutation, and
        %$\sigma(\cdot)$ is a permutation which maps the set of $n$ integers to itself bijectively, and 
        let $\mathfrak{S}_\numsrc$ be the set of all permutations of $\numsrc$ objects. 
        Then, the PIT loss $\mathcal{L}_\text{\normalfont PIT}$ is defined as
        \begin{equation}
            \mathcal{L}_\text{\normalfont PIT} \coloneqq 
            \mathbb{E}
                \bigg[
                    \frac{1}{\numsrc}\min_{\sigma \in \mathfrak{S}_\numsrc} \sum_{i=1}^\numsrc C_{i\sigma(i)} 
                \bigg],
            \label{eq:objective}
        \end{equation}
        where the expectation $\mathbb{E}[\cdot]$ is computed with respect to 
        the empirical distribution of the training data $($i.e.,\ the batch average$)$.
    \end{dfn*}
    The straightforward evaluation of ``$\min_{\sigma \in \mathfrak{S}_\numsrc}$,'' however, works only when $\numsrc$ is very small.
    If we evaluate it by a brute-force search
    (which is actually the case in some open implementations including the Asteroid~\cite{Pariente2020Asteroid}\footnote{
        In Oct.\ 27, 2020 (after the initial submission of this paper),
        Asteroid's default PIT was improved to $O(N^3)$
        using the Kuhn-Munkres (Hungarian) algorithm.
        The experiments in this paper are based on the previous factorial version.
    }),
    we need to scan all $|\mathfrak{S}_\numsrc| = \numsrc!$ permutations to find the best $\sigma$
    for \textit{every single step} of the stochastic gradient descent (SGD).
    Using the author's machine, each step takes about 4 seconds 
    when $\numsrc=10$.
    Assuming that $\text{10}^\text{6}$ SGD steps are required, it would take
    $\text{4}\times \text{10}^\text{6}$[s] = 46 days
    to perform a PIT-based training.
    Furthermore, as the factorial complexity $\bigo(\numsrc!)$ explodes very rapidly, 
    it becomes entirely impractical if $\numsrc$ increases only a bit more.

    This paper presents a practical solution to this problem. 
    In this paper, we propose a \textit{SinkPIT} (Sinkhorn PIT),
    a variant of the PIT loss.
    The alternative PIT loss telescopes our permutation space odyssey into an excursion,
    and enables us to try a large $\numsrc$ much larger than ever.
    The new loss is based on Sinkhorn's algorithm,
    which is recently attracting the interests in optimal transport and some fields of 
    machine learning~\cite{Moon_Gunther_Kupin_2009,Cuturi_2013,altschuler2017near,santa2017deeppermnet,Mena_Belanger_Linderman_Snoek_2018}.
    The contribution of this paper is threefold.
    (1) We propose a SinkPIT loss, a surrogate loss
    of the original PIT loss,
    and we experimentally verify that it works in our problem setting.
    (2)
    To our knowledge, this is the first attempt to decompose a single channel audio mixture into 10 speech signals which are densely overlapping.
    (3) To our knowledge, this is the first report on the effectiveness of Sinkhorn's algorithm 
    in the source separation problems.
    
\section{Sinkhorn PIT Loss}
    In this section, we derive a relaxed version of PIT, which is computationally tractable. Our goal is Eq.~\eqref{eq:sinkhornobjective}.
    Although many of the mathematical concepts in this section are well-known in applied optimal transport 
    (see e.g.~\cite{Cuturi_2013},\cite[\S 6.4]{santambrogio2015optimal},\cite[\S 7.3]{galichon2018optimal},\cite[\S 4]{peyre2019computational}),
    we summarize basic concepts for convenience.

\subsection{Doubly Stochastic Matrix and Sinkhorn's Theorem}
    A matrix $\mat{B} = (B_{ij})$ is \textit{doubly stochastic} if its all elements are non-negative and it marginally sums 1;
    i.e., $\sum_{i} B_{ij} = \sum_{j} B_{ij} = 1$.
    If all the elements of a doubly stochastic matrix are either $0$ or $1$,
    the matrix is called a permutation matrix.
    Let $\birkhoff$ denote the set of doubly stochastic matrices of size $\numsrc \times \numsrc$,
    which is often referred to as Birkhoff's polytope,
    and let $\perms$ denote the set of permutation matrices.
    Evidently $\perms \subset \birkhoff \subset \mathbb{R}^{\numsrc\times \numsrc}$, and moreover, 
    it is known that $\birkhoff$ is the smallest convex set that contains $\perms$
    and that $\perms$ is the set of the all vertices of the polytope $\birkhoff$.
    (Birkhoff-von~Neumann's theorem~\cite[Cor.11.5]{korte2012combinatorial}). 

    Since $\perms$ and $\mathfrak{S}_\numsrc$ are essentially the same 
    (let $P_{ij} = \delta_{\sigma(i)j}$),
    we may rewrite 
    Eq.~\eqref{eq:objective} using permutation matrices as follows, 
    \begin{equation}
        \mathcal{L}_\text{\normalfont PIT} 
        = \frac{1}{\numsrc} \mathbb{E} 
            \Big[
                    \min_{ \mat{P} \in \perms} \frob{ \mat{C} }{ \mat{P} } 
            \Big] \label{eq:frob}
    \end{equation}
    where $\frob{\mat{X}}{\mat{Y}} = \mathrm{Tr}(\mat{X}^\mathsf{T}\mat{Y}) = \sum_{ij}X_{ij}Y_{ij}$ 
    denotes the Frobenius inner product.
    We are interested in relaxing this quantity by using $\birkhoff$.
    As a preparation for that, let us introduce the celebrated Sinkhorn's theorem.
    \begin{theorem}[Sinkhorn~\cite{sinkhorn1964relationship}]\label{thm1} 
        Let $\mat{Y}$ be a square matrix whose all elements are strictly positive. 
        {\normalfont (1)} 
            There exist diagonal matrices $\mat{D}_1$ and $\mat{D}_2$
            such that $\mat{B} \coloneqq \mat{D}_1 \mat{Y} \mat{D}_2$ is doubly stochastic,
            and they are unique up to a multiplicative constant.
        {\normalfont (2)} 
            By rescaling the columns and rows alternately so that the sums are 1, 
            $\mat{Y}$ converges to the aforementioned matrix $\mat{B}$.
    \end{theorem}
    The second part of the theorem is sometimes referred to as the Sinkhorn iteration.
    Numerically, it is more convenient computing the iteration on a log scale.
    That is, the initial value is $\mat{Z}^{(0)} = \log \mat{Y}$ (element-wise log),
    and the Sinkhorn iteration translates to the alternating subtraction of the log-sum-exp as follows,
    \begin{alignat}{2}
    \textstyle    
        Z_{ij}^{(2l+1)} & 
            \gets Z_{ij}^{(2l)}   &&- \textstyle \logsumexp{i'} Z_{i'j}^{(2l)} \label{sinkhorniter1} \\
    \textstyle
        Z_{ij}^{(2l+2)} & 
            \gets Z_{ij}^{(2l+1)} &&- \textstyle \logsumexp{j'} Z_{ij'}^{(2l+1)}\label{sinkhorniter2}
    \end{alignat}
    then $\exp \mat{Z}^{(k)}$ (element-wise exp) converges to a doubly stochastic matrix 
    $\mat{B} = \mat{D}_1 \exp (\mat{Z}^{(0)}) \mat{D}_2$ as $k \to \infty$.

\subsection{Where the Sinkhorn Iteration Converges} 
    Now let us introduce a theorem that connects the objective Eq.~\eqref{eq:frob} and the limit of the Sinkhorn iteration.
    The following theorem states that the doubly stochastic matrix obtained by the Sinkhorn iteration 
    minimizes the quantity Eq.~\eqref{eqsoftperm}, which is an entropy-regularized version of Eq.~\eqref{eq:frob}.
    Hereafter, we add subscripts to $\mat{Z}^{(k)}$ to explicitly note that it is dependent on $\beta$ and $\mat{C}$.
    \begin{theorem}\label{thm2}
        Let $\mat{C} \in \mathbb{R}^{\numsrc\times \numsrc}$ be an arbitrary square matrix,
        $\beta > 0$ be an inverse temperature parameter,
        and let $\entropy{\mat{B}}= -\sum_{ij} B_{ij} \log B_{ij}$ be the entropy of $\mat{B}$.
        Then, following doubly stochastic matrices are equal.
        \begin{enumerate}[\normalfont (I)]
        \item 
            The limit of the Sinkhorn iteration $($Eqs.~\eqref{sinkhorniter1}, \eqref{sinkhorniter2}$)$
            $\exp \mat{Z}^{(\infty)}_{\beta, \mat{C}}$,
            where the initial value is $\mat{Z}^{(0)}_{\beta, \mat{C}} = -\beta \mat{C}$. 
        \item 
            $\hat{\mat{B}}_{\beta, \mat{C}} \in \birkhoff$ 
            which minimizes an entropy-regularized Frobenius inner product
            $\mathcal{S}_{\beta, \mat{C}}(\mat{B})$,
        defined by
        \end{enumerate}
        \begin{equation}
                \mathcal{S}_{\beta, \mat{C}}(\mat{B}) 
                    \coloneqq \frob{\mat{C}}{\mat{B}} - \frac{1}{\beta}\entropy{\mat{B}}
                    = \frob{   \mat{C} + \frac{1}{\beta} \log \mat{B}}{~\mat{B}}. \label{eqsoftperm} 
        \end{equation}
    \end{theorem}
    \begin{proof}
        \normalfont
            The matrix (I) has a form $\mat{D}_1 \exp(-\beta \mat{C}) \mat{D}_2$ ($\because$ Theorem 1(2)).
            A simple calculation based on the Lagrange multipliers method
            reveals that the matrix (II) can be written as $\hat{\mat{B}} = \mat{D}_3 \exp(-\beta \mat{C}) \mat{D}_4$.
            Because of the uniqueness (Theorem 1(1)), they are equal.
            See~\cite[Lem.2]{Cuturi_2013},~\cite[Lem.1]{Mena_Belanger_Linderman_Snoek_2018}~\cite[Prop.4.3]{peyre2019computational}.
    \end{proof}
    In addition, under some technical assumptions,
    it can be also shown that the cold limit $(\beta \to \infty)$ almost surely converges to a permutation matrix
    that minimizes our original objective Eq.~\eqref{eq:frob}~\cite{Mena_Belanger_Linderman_Snoek_2018}.

\begin{figure}[!t]
    \centering
    \includegraphics[width=\linewidth]{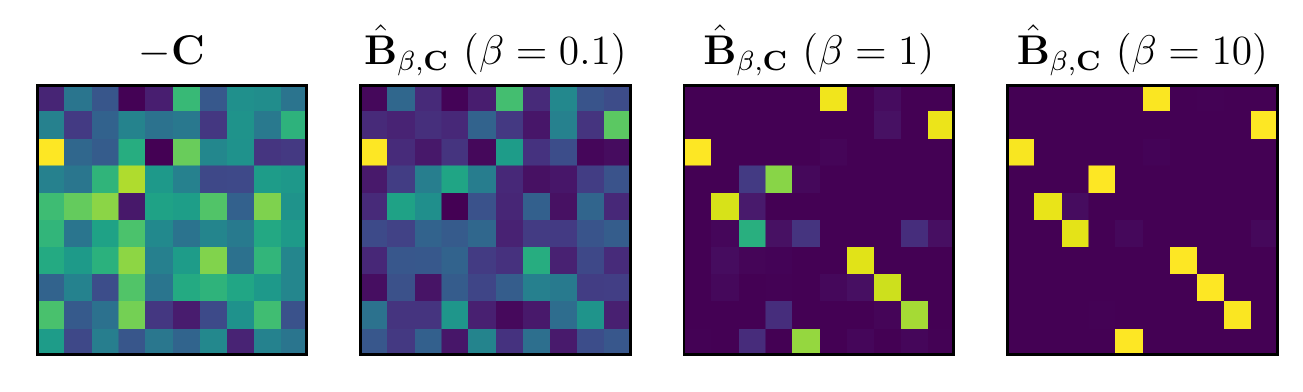}
    \caption{
        An example of a matrix $C_{ij} \sim \mathcal{N}(0, 10^2)$ 
        and the doubly stochastic matrices 
        $\hat{\mat{B}}_{\beta, \mat{C}}=\exp \mat{Z}_{\beta,\mat{C}}^{(k)}$,
        where $k=200$.
    }\label{fig:mat}
\end{figure}
\begin{figure}[!t]
    \centering
    \includegraphics[width=0.8\linewidth,clip]{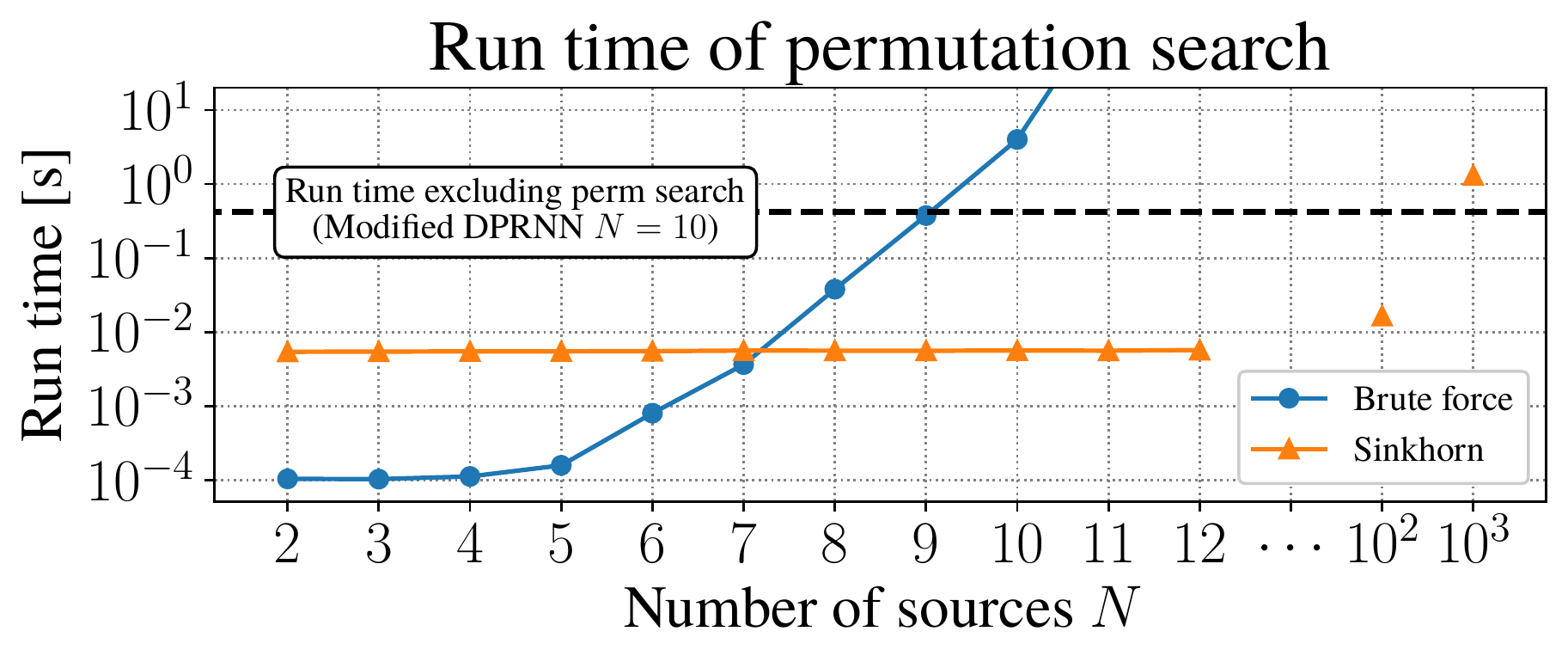}
    \caption{
        The computation time to solve a permutation problem once.
        The ``brute force'' is based on the Asteroid~\cite{Pariente2020Asteroid} implementation.
    }\label{fig:runtime}
\end{figure}

\subsection{Our New Objective: The SinkPIT Loss}
    Putting all things together, we may rewrite our objective Eq.~\eqref{eq:frob} as
    \begin{equation}
        \mathcal{L}_\text{\normalfont PIT} =
            \frac{1}{\numsrc} \mathbb{E}
                \big[
                    \lim_{\beta \to \infty} \lim_{k \to \infty} 
                    \mathcal{S}_{\beta, \mat{C}}(\exp \mat{Z}^{(k)}_{\beta,\mat{C}} )
                \big].
    \end{equation}
    Thus we have successfully removed the evaluation of ``$\min_{\mat{P} \in \perms}$''
    from our objective function.
    The ``$\to \infty$'' is still computationally intractable,
    but we can use some finite $k$ and $\beta$ instead.
    Fig.~\ref{fig:mat} shows 
    that small $k$ and $\beta$ work in practice.
    Thus we may define our objective function as follows, and
    let us call it a \textit{SinkPIT} loss.
    \begin{equation}
        \boxed{
            \mathcal{L}_\text{SinkPIT} \coloneqq 
                \frac{1}{\numsrc} \mathbb{E}\left[
                    \frob{\mat{C} + \frac{1}{\beta}\mat{Z}^{(k)}_{\beta, \mat{C}}~}{~  \exp \mat{Z}^{(k)}_{\beta, \mat{C}}  }
                \right].
        }
        \label{eq:sinkhornobjective}
    \end{equation}
    The computation complexity of Eq.~\eqref{eq:sinkhornobjective} is 
    $\bigo(k \numsrc^2) \ll \bigo(\numsrc !)$. 
    Fig.~\ref{fig:runtime} compares the actual run time.
    When $\numsrc \ge 8$, the SinkPIT is significantly faster than the ordinary PIT.
    Fig.~\ref{diagram:sinkpit} shows the computation graph of the SinkPIT training.
    The cost of the Sinkhorn iteration is typically much smaller than the other parts of the training,
    such as the backpropagation inside of the target model $f_\theta$.

\subsection{Relation to Other PIT Siblings}
    In this section, let us briefly consider the relationship between the SinkPIT loss and other PIT variants.
    To date, some PIT variants have been proposed, including 
    the OR-PIT~\cite{takahashi2019recursive} and the ProbPIT~\cite{yousefi2019probabilistic},
    the former being a greedy approximation,
    and the latter being a probabilistic relaxation of the ordinary PIT.
    The SinkPIT is particularly close to the ProbPIT,
    which is defined as follows.
    \begin{align}
    \textstyle
        \mathcal{L}_\text{ProbPIT} 
        \coloneqq& \textstyle 
        \mathbb{E} [ - \log \sum_{\sigma \in \mathfrak{S}_\numsrc} \prob(\sigma)^{1/\gamma} \prod_{i=1}^{\numsrc} \prob(y_i, s_{\sigma(i)})^{1/\gamma} ] \\
        =& \textstyle 
        \mathbb{E} [ - \log \sum_{\mat{P} \in \perms} \prob(\mat{P})^{1/\gamma} \exp (-\frob{\mat{C}}{\mat{P}}/\gamma ) ]
    \end{align}
    where $p(y_i,s_j) = \exp(-C_{ij})$, and
    $p(\sigma)$ is a prior distribution over the permutations $\mathfrak{S}_\numsrc$. 
    Using the ``addition'' $\oplus_\gamma$ and ``multiplication'' $\otimes_\gamma$ operators of the log semiring,
    defined by the relation
    $x \circledast_\gamma y \coloneqq - \gamma \log (e^{-x/\gamma} \ast e^{-y/\gamma})$
    where $\ast \in \{ +, \times\}$
    (see e.g.~\cite{goodman1999semiring,mohri2002weighted}),
    it is written in a ``sum-of-product'' form as follows,
    \begin{equation}
            \mathcal{L}_\text{ProbPIT} 
                \propto 
            \mathbb{E}
                \bigg[
                    \mybigoplus_{\mat{P} \in \perms} \lambda^{}_{\scriptscriptstyle \mat{P}} \otimes_\gamma \frob{\mat{C}}{\mat{P}}
                \bigg],
                \label{eq:probpit}
    \end{equation}
    where $\lambda^{}_{\scriptscriptstyle \mat{P}} \coloneqq - \log \prob(\mat{P})$.
    Noting that $x \oplus_\gamma y \to \min(x, y)$ at the tropical limit $\gamma \to 0$,
    and that $\otimes_\gamma$ is always the same as $+$,
    the ProbPIT goes to $\mathbb{E}[\min_{\mat{P} \in \perms} (\lambda^{}_{\scriptscriptstyle \mat{P}} + \frob{\mat{C}}{\mat{P}})]$
    which is essentially the same as the ordinary PIT if the prior is sufficiently flat.

    On the other hand, the predominant term of our SinkPIT is also written in 
    a sum-of-product form as follows,
    \begin{equation}
        \mathcal{L}_\text{SinkPIT} \appropto
        \mathbb{E} [ \langle {\mat{C}}, {\hat{\mat{B}}_{\beta, \mat{C}}} \rangle_\mathsf{F} ]
         = \mathbb{E}
            \bigg[ 
                \sum_{\mat{P} \in \perms} \mu^{}_{\scriptscriptstyle \mat{P}} \times \frob{\mat{C}}{\mat{P}} 
            \bigg],
        \label{eq:expansion}
    \end{equation}
    because of Birkhoff-von Neumann's theorem 
    which states that $\mat{B} \in \birkhoff$ can always be written as 
    $\mat{B} = \sum_{\mat{P}\in \perms} \mu^{}_{\scriptscriptstyle \mat{P}} \mat{P}$
    (where $\mu^{}_{\scriptscriptstyle \mat{P}} \ge 0, \sum \mu^{}_{\scriptscriptstyle \mat{P}} = 1$),
    as well as the linearity of $\frob{\cdot}{\cdot}$.

    Although the algebraic operators are different,
    both PIT variants are 
    understood as the ``weighted sums'' of the losses over the hypothetical permutations.
    In addition, both losses are parameterized relaxations of the original PIT,
    and they go to the original objective at the limits. 
    However, their computation complexities are different,
    as the ProbPIT requires the summation over $\numsrc!$ terms unless the prior is sparse,
    whereas the SinkPIT skips such a costly computation
    by the aforementioned efficient algorithm.

\section{Neural Network for 10 Speaker Separation}
        This section describes the details of a neural network model to separate a signal into many components.
        (We consider the cases $\numsrc=5$ and $\numsrc=10$.)
        Firstly, let us briefly introduce the TasNet framework~\cite{luo2018tasnet,Luo_2019,luo2020dual,chen2020dual};
        the TasNet architecture (Fig.~\ref{fig:tasnet}) 
        is simply a time-frequency (T-F) masking on a learned T-F representation,
        instead of the conventional STFT domain.
        The main difference of the TasNet models is what neural network models to use as a MaskNet.
        For example, the ConvTasNet~\cite{Luo_2019} uses the convolutional networks,
        and the DPRNN~\cite{luo2020dual} uses the recurrent networks.

        In our experiment, we used the DPRNN~\cite{luo2020dual} as the base model,
        and we modified it a little.
        Although we have not exhaustively explored the hyperparameters, 
        the following modifications seemed to be suitable for our problem.
        \begin{enumerate}[(i)]
            \item 
                We used a longer Encoder/Decoder window (original: $l_\text{win} =2$,  modified: $l_\text{win} =16$),
                and more filters (orig: $d_\text{filter} = 64$, mod: $d_\text{filter} = 256$).
            \item 
                When $\numsrc=10$, we increased the dimension of the RNN hidden units 
                (orig: $d_\text{RNN}=128$, mod: $d_\text{RNN}=384$),
                and the number of DPRNN blocks (orig: 6 blocks, mod: 8 blocks) 
                (We did not modify these parameters when $\numsrc=5$.)
            \item 
                We replaced the gated activation after the 2D convolution,
                $\text{tanh}(\cdot)\cdot \text{sigmoid}(\cdot)$,
                with $\text{PReLU} (\cdot) \cdot \text{sigmoid}(\cdot)$.
            \item 
                Instead of directly expanding the dimension from $d_\text{RNN}$ to $\numsrc d_\text{filter}$
                by a single affine layer at the last of the MaskNet,
                we used a small point-wise network: 
                $
                \text{Affine} ({d_\text{RNN}} \to {4 \numsrc d_\text{filter}})
                \to 
                \text{PReLU} 
                \to
                \text{Affine}({4 \numsrc d_\text{filter}} \to {\numsrc d_\text{filter}})
                \to 
                \text{softmax}$.
        \end{enumerate}

    Let us specify the parameters of the SinkPIT, $\mat{C}, \beta$ and $k$.
    The base pairwise loss matrix $\mat{C}$ was the negative value of SI-SDR~\cite{le2019sdr}\footnote{
        Some readers may be tempted to 
        use the SI-SDRi, i.e.\ $C_{ij} = -(\text{SI-SDR}(y_j,s_i) - \text{SI-SDR}(x,s_i))$.
        However, it will not improve the performance 
        because the result of the Sinkhorn iteration is invariant under the transform
        $e^{-\beta C_{ij}}
        \mapsto 
        e^{-\beta C_{ij} + \xi_i + \eta_j}$
        for any vectors $\bm{\xi}, \bm{\eta} \in \mathbb{R}^N$.
    },
    i.e., $C_{ij} = - \text{SI-SDR}(y_j(t) ,s_i(t) )$, 
    where SI-SDR is defined as
    \begin{equation}
        \sisdr(u(t), v(t)) = 10 \log_{10} \frac{\langle u, v\rangle^2 }{\|u\|^2\|v\|^2 - \langle u, v\rangle^2 }.
    \end{equation}
    During the training, we annealed the coldness parameter $\beta$.
    The cooling schedule was
    $
        \beta = \min(1.02^{\#\text{epoch}}, 10).
    $
    The number of iteration was always $k=200$.

\begin{figure}[!t]
        \centering
        \fbox{
        \includegraphics[width=0.9\linewidth,clip]{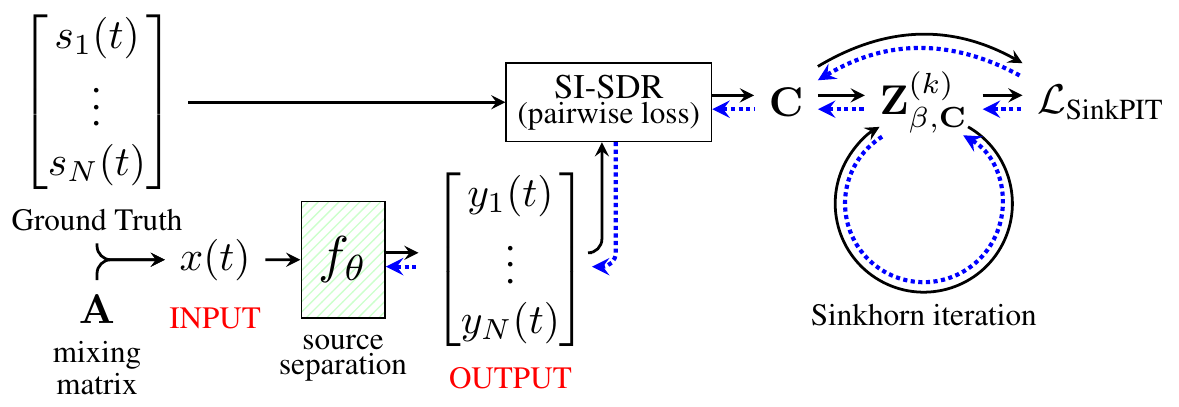}%
        }
        \caption{
            The computation graph of the SinkPIT-based training. 
            Blue dotted arrows indicate the backpropagation paths.}
        \label{diagram:sinkpit}
\end{figure}
\begin{figure}[!t]
        \centering
        \fbox{
        \includegraphics[width=0.9\linewidth,clip]{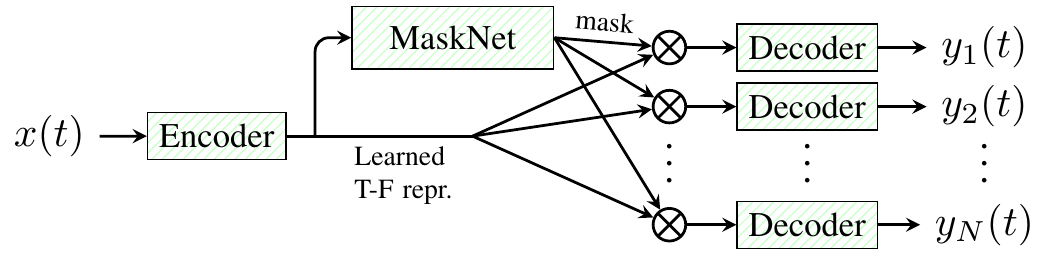}%
        }
        \caption{The computation graph of the TasNet framework.}
        \label{fig:tasnet}
\end{figure}

\begin{figure*}[!t]
    \centering
    \begin{subfigure}[t]{1\textwidth}
        \centering
            \includegraphics[width=\linewidth]{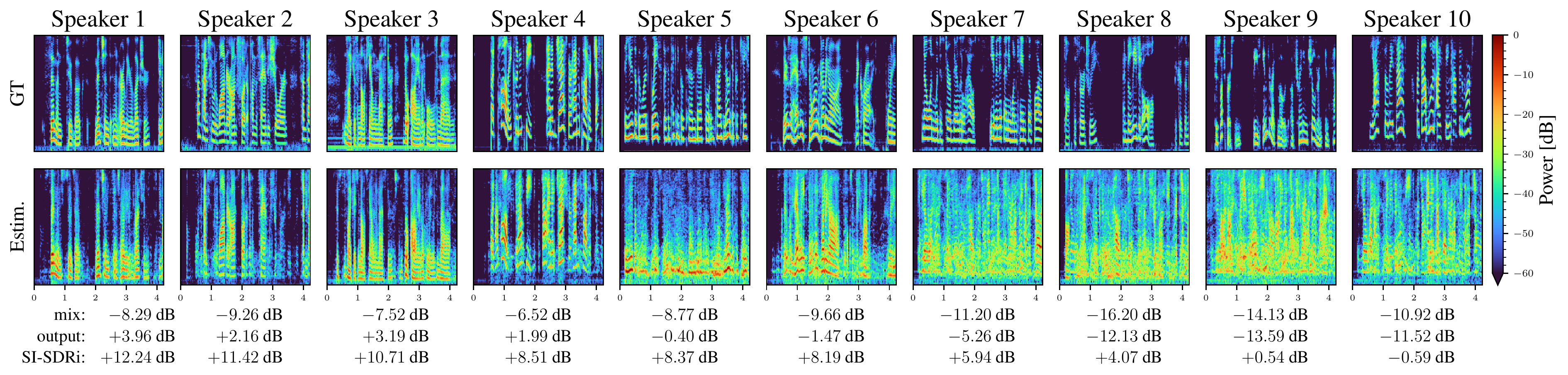}
            \caption{
                A typical result of 10-mix separation (mean SI-SDRi = 6.94 dB).
                1st row: mel spectrograms of ground truth signals (excerpt from LibriSpeech~\cite{panayotov2015librispeech}). 
                2nd row: mel spectrograms of estimated signals.
                For visility, the ground truths and the estimations are aligned, and they are sorted by the SI-SDRi.
            }\label{fig:specs}
    \end{subfigure}
    \\
    \begin{subfigure}[t]{0.26\textwidth}
        \includegraphics[width=1.0\linewidth]{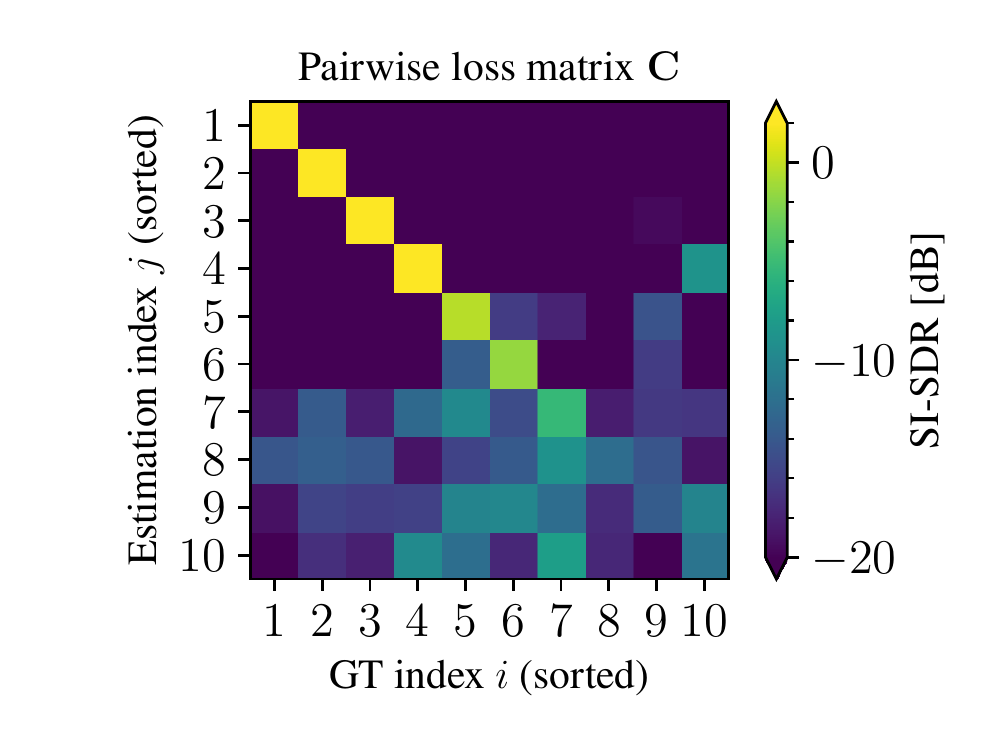}
        \caption{
            The pairwise loss matrix $\mat{C}$.\\The same test data as Fig.~\ref{fig:specs}.
        }\label{fig:pwloss}
    \end{subfigure}
    \begin{subfigure}[t]{0.39\textwidth}
        \includegraphics[width=1.0\linewidth]{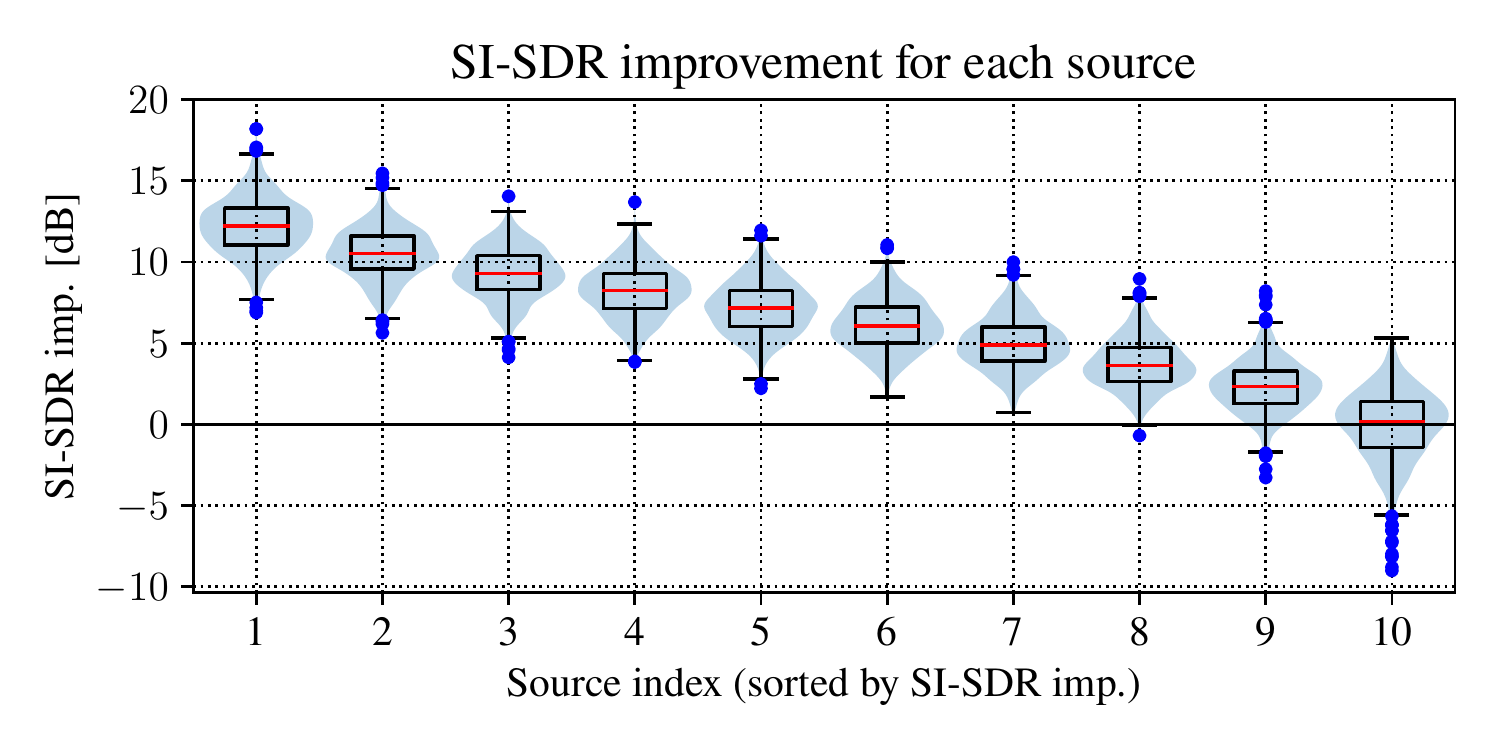}
        \caption{
            The distribution of SI-SDRi of the separated signals for each mixture. Sorted by SI-SDRi.
        }\label{fig:rank}
    \end{subfigure}
    \begin{subfigure}[t]{0.33\textwidth}
        \includegraphics[width=1.0\linewidth]{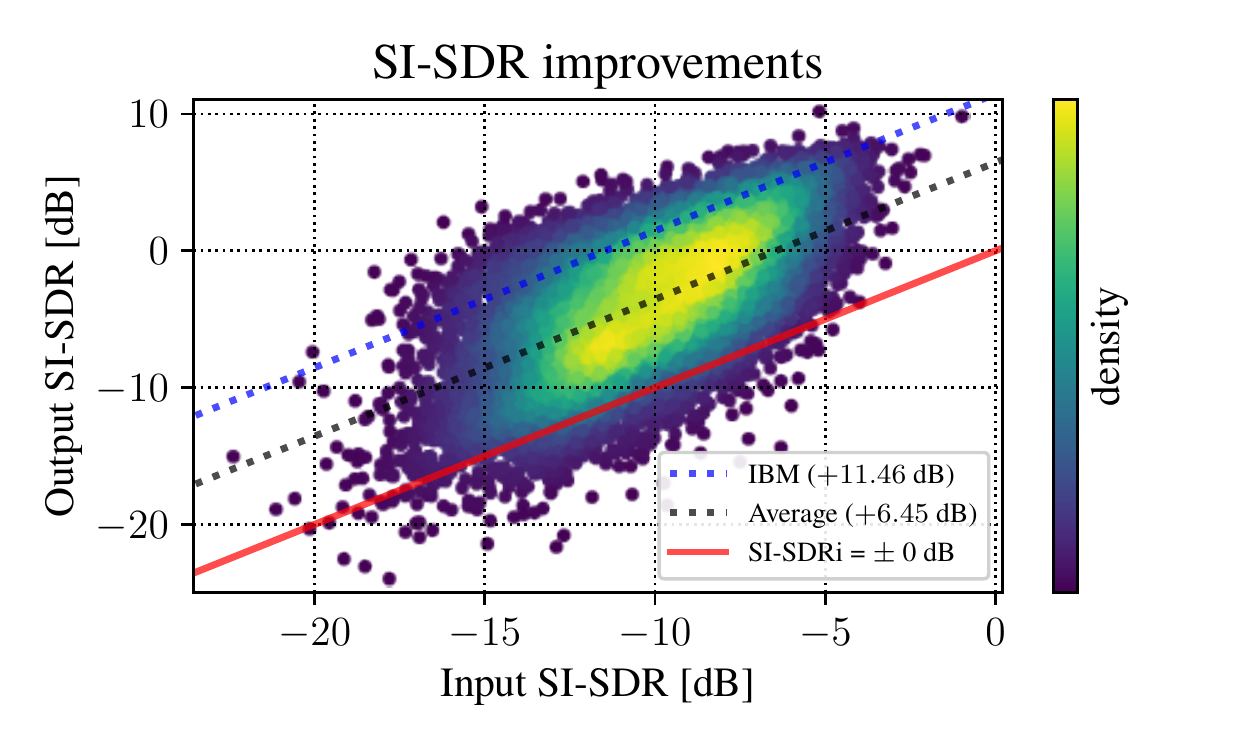}
        \caption{
            The distribution of the input/output SI-SDRs.
        }\label{fig:scatter}
    \end{subfigure}
    \caption{Experimental results where the number of source signals was $\numsrc=10$.}\label{fig:expresults}
\end{figure*}

\section{Experiment}

\subsection{Experimental Setup}
    We experimentally evaluated the effectiveness of the SinkPIT
    based on a standard evaluation protocol provided in the Asteroid framework~\cite{Pariente2020Asteroid}.
    We basically used the default hyperparameters of the Asteroid, 
    unless otherwise noted (see also the previous section).
    The evaluation metric was the SI-SDRi (SI-SDR improvement).

    The dataset we used was the LibriMix~\cite{panayotov2015librispeech,cosentino2020librimix}.
    The configuration was ``$N$mix-8kHz-min-clean'':
    the number of source signals was $N (=5, 10)$,
    the sampling rate was 8~kHz, 
    each source in a mixture was adjusted to the shortest one (``min''),
    and there was no noise in the mixture (``clean'').
    The original LibriMix provides only 2-mix and 3-mix data, but 
    we generated 5-mix and 10-mix data using a LibriMix script.
    The numbers of the train/valid/test data were 
    $\textit{train}=20300,\textit{valid}=3000, \textit{test}=3000$ (5-mix) and 
    $\textit{train}=10100,\textit{valid}=1000, \textit{test}=1000$ (10-mix).

    To train the model, we used the Adam optimizer ($\beta_1=0.5,\beta_2=0.9,\varepsilon=10^{-3}$).
    The learning rate started at $\mathit{lr}=10^{-3}$, 
    and was halved if the learning stagnated in the last 5 epochs,
    with the minimum value of $\mathit{lr} = 5\times 10^{-5}$.
    The batch size was 8,
    4 GPUs were used,
    and the number of training epochs was 300.
    During the training, 
    unlike the LibriMix practice of using a fixed mixing matrix 
    $\mat{A} \in \mathbb{R}_{>0}^{1 \times \numsrc}$ for each mixture,
    we randomly drew a different $\mat{A}$ from a log-uniform distribution (10 dB of dynamic range) 
    at each step for data augmentation.

\begin{table}[!t]
    \centering
    \caption{SI-SDR improvements}\label{tab:onlytable}
    {\small
    \begin{tabular}{llrr}
        \hline
        Method & PIT &  5-mix & 10-mix  \\ 
        \hline\hline
        Ideal Binary Mask & -- & 11.35 dB & 11.46 dB   \\ 
        Ideal Ratio Mask & -- & 10.63 dB & 10.79 dB  \\ 
        \hline
        Modified DPRNN & PIT      &  9.26 dB & N/A  \\
        Modified DPRNN & ProbPIT{\scriptsize ${(\gamma=0.1)}$}   & 9.60 dB & N/A  \\
        Modified DPRNN & SinkPIT   &  9.39 dB & 6.45 dB \\
        \hline
    \end{tabular}
    }
\end{table}

\subsection{Results}
    Table~\ref{tab:onlytable} shows the SI-SDR improvement for each condition.
    Some audio samples (including the ones based on~\cite{takamichi2019jvs}) are also available at the author's website\footnote{
        \texttt{https://tachi-hi.github.io/research/}
    }.
    Regarding the run time of training, the following results were obtained.
    For $\numsrc=5$, there was no significant difference,
    with each epoch requiring approximately 13~min (PIT and ProbPIT), and 15~min (SinkPIT), respectively.
    For $\numsrc=10$, on the other hand, 
    a single epoch required about 3~hr~30~min (PIT and ProbPIT) and 17~min (SinkPIT), respectively.

    Fig.~\ref{fig:expresults} shows more detailed results for the case of $\numsrc=10$.
    Although the source separation from a single channel to 10 sources is a very difficult problem, 
    the strong correlation between the estimated signals and the ground truth sources 
    can be observed visually, and it can also be confirmed quantitatively through SI-SDRi values (Fig.~\ref{fig:specs}).
    Fig.~\ref{fig:pwloss} shows a pairwise loss $\mat{C}$,
    in which we may observe that some sources are clearly separated from others,
    while some are still mixed.
    Fig.~\ref{fig:rank} shows 
    the SI-SDRi of each source in a mixture.
    The SI-SDRi of the best source in the mixture often achieves $> +10$~dB,
    and the top 5 sources almost always achieve $> +5$~dB.
    All but the worst source almost always achieve $> +0$~dB.
    Fig.~\ref{fig:scatter} shows the relation between the input and the output SI-SDR.

\section{Conclusion and Future Work}
    This paper proposed the SinkPIT, a variant of the PIT loss,
    which is incomparably faster than the ordinary PIT if the number of sources is large.
    It enabled us to train a many-source separation model in a reasonable time.
    Although the separation quality was not yet perfect,
    we would say the experimental results were still promising enough.
    It is not at a practical level at present
    if the number of source signals is $\numsrc=10$, 
    but some improvement can be expected with a simple combination of the basic ideas.
    For instance, a multi-channel setting would be interesting.
    By using several microphones ($\nummix \ge 2$), 
    it may be enabled to extract 10 or more speakers' voices more clearly.

    These positive results encourage us to tackle the very large scale source separation problems in earnest.
    It will make a starting point of a new direction of audio source separation studies,
    as the idea of SinkPIT is almost ``perpendicular'' to other concepts of source separation techniques.
    There seem to be a huge amount of research topics worthy of consideration, 
    even within straightforward ideas.
    These are subjects for future study.

\bibliographystyle{ieeetr}
{
    \bibliography{myref}
}

\end{document}